\documentclass[8.5pt]{article}
\usepackage{setspace}
\usepackage{multicol}
\usepackage{amsmath}
\usepackage{amssymb}
\usepackage{latexsym} 
\usepackage[all]{xypic}
\usepackage{listings}
\usepackage{amsmath,amsfonts,amssymb}
\usepackage{mathrsfs}
\usepackage{amsmath}
\usepackage{amssymb}
\usepackage{enumerate}
\usepackage{eufrak}
\usepackage{fancyhdr}
\usepackage[dvips,breaklinks=true]{hyperref}
\usepackage[all]{xypic}
\usepackage{mytheorems}

  \title{The Star Height Hierarchy Vs. The Variable Hierarchy}
  \author{Walid Belkhir \\
\begin{small} Laboratoire d'Informatique Fondamentale de Marseille, France\end{small}}
   \date{}
\input macros.tex
\begin{document}
\maketitle


\begin{abstract} The star height hierarchy (resp. the variable hierarchy) results in classifying $\mu$-terms into classes according to the nested depth of fixed point operators (resp. to the number of bound variables).   We prove, under some assumptions, that   the variable hierarchy is a proper refinement of the star height hierarchy. We mean that the non collapse of the variable hierarchy implies the non collapse of the star height hierarchy. The proof relies on the combinatorial  characterization of the two hierarchies.
\end{abstract}

 \section{Introduction}
Roughly speaking,  a  $\mu$-calculus in the abstract sense of Arnold and  Niwi{\'n}ski  \cite[\S 2]{AN} is a set of syntactic entities and a set of formal operations. The latter consists in the  fixed point operators $\mu$ and $\nu$ and the substitution operation.  These syntactic entities come with and intended interpretation over a class of complete lattices. Each entity $t$ is interpreted as a monotonic mapping from $t^{ar(t)}$ to $L$, where $ar(t)$ is the arity of $t$, that corresponds to  the  free variables of  $t$, and $L$ is a complete lattice. The entities $\mu x .t$ and $\nu x. t$  of a $\mu$-calculus are interpreted respectively as the least and greatest parametrized fixed point  of the interpretation of $t$. The substitution is interpreted by means of the functional  composition.

Hierarchies and logical expressiveness are at the core of fixed point theory.   The most  known  and well studied measure of the complexity  of the  $\mu$-calculi is the \emph{alternation depth} of its $\mu$-terms \cite{clones,Bradfield98,ANhier}, that is, the number of alternations between $\mu$ and $\nu$.  As a consequence,  it is possible  to construct a  hierarchy  of  $\mu$-terms  according to the alternation depth measure.  \longversion{The first level of this hierarchy contains  for example PDL, CTL \dots We mean that one can encode PDL and CTL formulae with modal $\mu$-formulae of alternation depth one.  This is not the case of  Parikh's Game Logic GL \cite{Parikh}, since  GL intersects with all the classes of the hierarchy \cite{BerGamelogic03}. As a consequence,  it is not possible to use the alternation depth hierarchy to prove  that   GL  is less expressive than  $L_{\mu}$.  Another option was pursed in \cite{BerwangerGraLen06} by encoding  GL formulae with $\mu$-formulae of just two fixed point variables.  This encoding gave rise to  an orthogonal hierarchy, that is the hierarchy of $L_{\mu}$ formulae induced by the number of fixed point variables.  The non collapse of the variable hierarchy \cite{BerwangerGraLen06} allowed to  separate GL from $L_{\mu}$. \\}
The \emph{variable hierarchy} \cite{BerwangerGraLen06,LPAR08} results in classifying  $\mu$-terms  into classes according to the number of bound variables i.e. fixed point variables. \longversion{Historically, the variable hierarchy problem was asked  by Immerman and Poizat \cite{Immerman} in the context of  back and forth games $\homl A, B   \homr$.  The aim was of deciding whether   a logical formulae of fixed number of   bound variables is  able or not to distinguish the two structures $A$ and $B$. Or, in other words,  they were  asking whether the two structures $A$ and $B$ are models of the same logic formulae where the number of bound variables is bounded.}  By introducing the variable hierarchy for the propositional modal $\mu$-calculus \cite{kozen}  and showing that it does not collapse \cite{BerwangerGraLen06}, the authors managed to separate  Parikh's Game Logic  \cite{Parikh} from  the modal $\mu$-calculus and solve a long standing open problem.      
   
The third hierarchy is the \emph{star height hierarchy} \cite{EgganTransitionGraphs,CourcelleStarHieight}: the $\mu$-terms are classified into levels of a hierarchy according to the nested depth of the application of  fixed point operators (or the iteration operator).   The star height problem was first  asked in formal language theory and consists in answering  whether all regular languages can be expressed using regular expressions of bounded  star height. This question have been answered by Eggan in \cite{EgganTransitionGraphs}, where he gave examples of regular languages of star height $n$ for every $n \in \mathbb{N}$.  
The  star height problem  was asked later for  regular trees \cite{CourcelleStarHieight}, the latter  are finite or infinite trees with only finite many distinct subtrees, up to isomorphism. Regular  trees form the free iteration  theory and they might be  written by means of iterative   theory expressions. These expressions use an iteration operator, denoted  $\dag$,  which is  interdefinable   with the      Kleene's $*$  operator  in the case of matrix iteration theories \cite[\S 9]{bloomesik}.  The star height problem can be asked in a general way for iteration theories \cite{bloomesik}, where the \emph{dagger} $\dag$ operator is considered.\cutout {as well as for $\mu$-calculi, if  just one fixed operator among $\set{\mu,\nu}$ is considered.}

The alternation depth hierarchy and the variable hierarchy are orthogonal, however the variable hierarchy and the star height hierarchy are intuitively close. In this paper we show that the variable hierarchy is a proper \emph{refinement} of the star height hierarchy. That is, the non collapse of the former implies the non collapse of the latter.     The key observation is that the   combinatorial measure which   characterizes  the variable hierarchy  (i.e. the entanglement) is lower than the combinatorial measure  which   characterizes the star height hierarchy (i.e. the rank).

\longversion{ This chapter is devoted to compare, under some assumptions,  the star height hierarchy and the variables hierarchy.  We shall discuss at the end of this Chapter how the strictness of the variable hierarchy would be of help to deduce the strictness of the star height hierarchy.   The key Lemma is that the   combinatorial measure  characterizing the variable hierarchy  (the entanglement) is lower than the combinatorial measure  characterizing the star height hierarchy (the rank). }

\paragraph{Preliminaries and  Notations} ~\\
A digraph $G=(V_G,E_G)$ is a set  of vertices $V_G$ and a binary relation $E_G \subseteq V_G \times V_G$.  $G$ is \emph{strongly connected} if for each two vertices $v_1,v_2 \in V_G$ there  exists a path in $G$ from  $v_1$ to $v_2$. A \emph{strongly connected component}  of $G$ is a maximal strongly connected subgraph of $G$.  We shall write $scc(v)$ for the strongly connected component of  $G$ that contains the vertex $v$. A strongly connected component is \emph{trivial} if it reduces to a single vertex without loops.  
We shall write   $\scc{G}$ for  the set of the non trivial strongly connected components of $G$. We define a transitive relation  $\prec_G$ on $\scc{G}$ as follows: $G_1\prec_G G_2$  if and only if $G_1\neq G_2$ and  there is a path in $G$ from a vertex of $G_1$ to a vertex of  $G_2$.  If $G_1\in \scc{G}$ then let  $G_{1}^{\succ}=\set{G_2 \in \scc{G}  \;|\; G_1 \prec_G G_2  }$. If no confusion will arise then we shall write $\prec$ instead of $\prec_G$.

\section{The $\mu$-calculi: syntax and semantics}\label{mu:calc:AN:sec}
\longversion{Many syntactic entities\marginpar{Who cares?} can be structured to have a shape of a $\mu$-calculus. For instance, this happened to infinite words \cite[\S 5]{AN},  automata \cite[\S 7]{AN}, and  parity games \cite{TCS333}. The interpretation of an automaton, viewed as an entity of the $\mu$-calculus, is the language which it accepts.} We recall the definition of a $\mu$-calculus as given in  \cite[\S 2]{AN}. Let $E$ be a set of objects or entities and  let $Var$ be a fixed countable set of variables. The variables in $Var$ will be denoted by $x,y,z, \dots$ A mapping $\rho: Var \to E$  is called a \emph{substitution}. If $\rho$ is a substitution into some set $E$, $x$ a variable, and $e$ and element of $E$, we denote by $\rho\set{e/x}$ the substitution $\rho'$  defined by $\rho'(x)=e$ and $\rho'(y)=\rho(y)$ if $y \neq x$. More generally, if $x_1,\dots,x_n$ are \emph{distinct} variables and if $e_1,\dots,e_n$ are elements of $E$, then
$
\rho\set{e_1/x_1,\dots,e_n/x_n}
$ 
is the substitution $\rho'$ defined by
\begin{small}
\begin{align*}
\rho'(y)=\left\{ \begin{array}{ll}
e_i \;\; & \tif y \in \set{x_1,\dots,x_n}, \\
\rho(y) & \tif y \notin \set{x_1,\dots,x_n}
\end{array}
\right.
\end{align*}  
\end{small}
 \begin{definition}\label{Mu-Calc-Def}
\emph{A $\mu$-calculus is a tuple $\homl T,id,ar,comp,\mu,\nu \homr$, where  
\begin{itemize}
\item $T$ is an arbitrary set, its elements are the \emph{$\mu$-terms} of the $\mu$-calculus.
\item $id$ is a mapping from $Var$ to $T$. We denote by $\hat{x}$ the element $id(x)$ in $T$.
\item $ar$ is a mapping associating to each $t\in T$ a subset of $Var$ called the \emph{arity} of $t$. If $x\in ar(t)$, we say that $x$ occurs \emph{free} in $t$, and the elements of $ar(t)$ are called the \emph{free variables} of $t$.
\item $comp$ is a mapping associating a $\mu$-term $comp(t,\rho)$ with any $\mu$-term $t$ and any substitution $\rho$; we shall write also $t[\rho]$.
\item $\mu$ and $\nu$ are two mappings from $Var\times T$ to $T$, the value of the mapping $\theta$ on $x$ and $t$ is written $\theta x.t$, for $\theta=\mu,\nu$. 
\end{itemize}
\longversion{Moreover, we assume the following axioms:
\begin{enumerate}
\item $ar(\hat{x})=\set{x}$,
\item $ar(t[\rho])= ar'(t,\rho)$, where  $ar'(t,\rho)=\bigcup_{y \in ar(t)} ar(\rho(y))$,
\item $ar(\theta x . t)=ar(t) \setminus \set{x}$,
\item $\hat{x}[\rho]=\rho(x)$, for $x \in Var$,
\item $t[\rho]= t[\rho']$ if $\rho | ar(t)= \rho' | ar(t)$,
\item $(t[\rho])[\pi]= t[\rho \star \pi]$, where $\rho \star \pi$ is the substitution defined by $\rho \star \pi(x)=\rho(x)[\pi]$,
\item if $ar'(\theta x . t,\rho) \neq Var$, there exists a variable $y\notin ar'(\theta. t,\rho)$ (possibly equal to $x$), such that $(\theta x . t)[\rho] = \theta y (t[ \rho\set{\hat{y} \backslash x}] )$.
\end{enumerate}
}
}
\end{definition}
Moreover, a $\mu$-calculus should satisfy further axioms, see \cite[\S 2]{AN}.
\paragraph{Semantics} Let $\homl T,id,ar,comp,\mu,\nu \homr$ be a $\mu$-calculus.  A $\mu$-\emph{interpretation} of    $T$ is a pair $(L,I)$ where $L$ is a complete lattice and $I$ is a function that associates to each  $\mu$-term $t$ a monotonic mapping $t: L^{ar(t)}\to L$  such that the substitution is interpreted as the functional composition and $\mu x. t$ (resp. $\nu x .t$) is interpreted as the least (resp. greatest) parametrized fixed point of the interpretation of $t$.   

Many syntactic entities can be structured to have a shape of a $\mu$-calculus. For instance, this happened to infinite words \cite[\S 5]{AN},  automata \cite[\S 7]{AN}, and  parity games \cite{TCS333}. The interpretation of an automaton, viewed as an entity of the $\mu$-calculus, is the language which it accepts.

\section{The star height  hierarchy}
\longversion{ In this \marginpar{Mieu expliquer rank/starheight} section we focus on the combinatorial part of the star height: the solution of the star height problem  requires   a digraph complexity measure called the \emph{rank}. The variable hierarchy problem  requires also for another digraph complexity measure, that will be discussed in Section \ref{Vari-Hierarchy-Section}. The comparison between the two hierarchies will be given in terms of  the comparison of their related digraph measures. } 
   
\begin{definition}
\emph{Let $t$ be a $\mu$-term and $Comp_0$ be the set if $\mu$-terms without application of fixed point operators $\mu$ and $\nu$. The star height of $t$ is defined as follows:
\begin{small}
\begin{align*}
h(t)=\left\{ \begin{array}{ll}
0 &\tif t\in Comp_0\\
Max\set{h(t'),h(\rho(x_1)),\dots,h(\rho(x_n)) } & \tif t=comp(t',\rho) \textrm{ where } x_i \in ar(t')\\
1+ h(t') & \tif t=\theta x. t' \textrm{ where } \theta=\mu,\nu
 \end{array}
\right.
\end{align*}  
\end{small}}
\end{definition}

\paragraph{The \emph{rank} : a digraph measure  for the star height} ~\\
In \cite{EgganTransitionGraphs} Eggan defined a complexity measure of digraphs, called the \emph{feed back number},  that captures the minimal  star height of regular languages. The minimal star height of a   regular language is exactly the feed back number of the minimal digraph of the expressions defining this  language. This measure has been formulated in a more natural way by Courcelle et al. in \cite{CourcelleStarHieight}, and they rename  it the \emph{rank}. There, they  solved the star height problem for regular trees, and showed that the minimal star height of a regular tree is exactly the rank of the minimal digraph of the tree.  
\begin{definition}
\emph{The  \emph{rank}  of a digraph $G$ is defined as follows:\\
$\bullet$ if $\scc{G}=\emptyset$, then $\rnk{G}=0$,\\
$\bullet$ if $\scc{G}=\set{G},$ then $\rnk{G}=1+ Min\set{r(G\setminus v) \;|\; v \in V_G)}$,\\
$\bullet$ otherwise, $\rnk{G}=Max\set{r(G')\;|\; G' \in \scc{G}}$.
}
\end{definition}
\noindent Note that $r(G)=0$ if $G$ is acyclic. If  $G$ is strongly connected and $r(G)=1$   then  $G$ contains a vertex whose removal makes the digraph acyclic.  It is not hard to  argue that  the rank of an  undirected path on  $n$ vertices   is $\lfloor log(n) \rfloor$.  
\paragraph{Thief and Cops games for the rank} ~\\
To establish the relation  between the rank and the entanglement,  in a first step, we  rephrase the definition of the rank in terms of games and strategies in the most direct way. 

\begin{definition}\label{Rank-Game-Def}
\emph{
The rank game $\Grnk{G,k}, k\ge 0$ played  alternatively between a  Thief and Cops  on the digraph $G$ is defined as follows.
\begin{itemize}
\item Its positions are of the form $(G',P,n)$ where $0\le n\le k$,  $G'$ is a subgraph of $G$, and  $P\in \set{Thief,Cops}$ such that 
\begin{itemize}
\item the starting position is  $(G,Thief,k)$,
\item  if $\scc{G}=\emptyset$ or $n=0$ then the play halts.  
\end{itemize}
\item If $\scc{G}= \set{G_1,\dots,G_l}$,  (possibly $l=1$) then Thief chooses some $G_i$ and  moves from $(G,Thief,n)$ to $(G_i,Cops,n)$.
\item If the position is $(G,Cops,n)$ then\footnote{Observe that in this case $G$ is strongly connected.}  Cops choose $v \in V_G$ and move    to  $(G\setminus v, Thief, n-1)$,
\item Thief  wins a play if and only if its   final\footnote{Observe that there is no infinite play.} position $(G',P,n)$ is such that $\scc{G}\neq \emptyset$ and $n=0$.
\end{itemize}
We define  $\Grnk{G}$ to be the minimum $k$ such that  Cops have a winning strategy  in the rank game $\Grnk{G,k}$.
}
\end{definition}

\begin{proposition}\label{Game1forRank:Lemma}
Let $G$ be  a digraph, then  $\Grnk{G}$ equals $\rnk{G}$.
\end{proposition}
\begin{proof}
We have just rephrased the definition of the rank by means of games and strategies following the game theoretic tradition.  That is,   Cops play the role of the \emph{minimizer}  and Thief plays the role of the \emph{maximizer}. 
\end{proof}

Now, in order to compare the rank with the entanglement in an easy way,  we shall give a useful variant of the   rank games. The idea is that, whenever $\scc{G}=\set{G_1,\dots, G_l}$ and Thief moves from $(G,Thief,n)$ to $(G_i,Cops,n)$, for some $i\in \set{1,\dots,l}$, then he is allowed later, and at any moment,  to \emph{come back} and  move  to $(G_j,Cops,n)$ where $G_i,G_j \in \scc{G}$ and $G_i \prec G_j$.

\begin{definition}\label{Rank2-Game-Def}
 Let $G$ be a digraph and $k\ge 0$. We define the rank game with \emph{come back}  $\grnk{G,k}$ between    Thief and  Cops   on the digraph $G$  as follows:\\
 Its positions are of the form $(G',P,L,n)$ where $0 \le n\le k$,  $G'$ is a subgraph of $G$,   $P\in \set{Thief,Cops}$, and $L$ is a set  of quadruplet of the form $(G,P,L,n)$  such that 
\begin{itemize}
\item the starting position is  $(G,Thief,\emptyset,k)$,
\item  if  ($\scc{G}=\emptyset$  or   $n=0$) and $L=\emptyset$  then the play halts.  
\end{itemize}
If $\scc{G}=\set{G_1,\dots,G_l}$,  (possibly $l=1$) then Thief has two kinds of moves: 
\begin{itemize}
\item he  chooses some $G_i\in \scc{G}$ and  moves from
\begin{align*}
(G,Thief,L,n) \to  (G_i,Cops, G_i^{\succ}\times (Cops,L,n) \cup L,n) \tag{forward move}
\end{align*}
where $\set{G_1,\dots,G_l} \times (Cops,L,n)=_{def}\set{(G_1,Cops,L,n),\dots,(G_l,Cops,L,n)}$.
\item  or he moves from 
\begin{align*} 
(G',Thief,L,n) \to B \textrm{ where } B\in L  \hspace{2cm} \tag{come back move}
\end{align*}
\end{itemize}
 If the position is $(G,Cops,n)$ then\footnote{Observe that  in this case $G$ is strongly connected.}  Cops choose $v \in V_G$ and move   to  $(G\setminus v, Thief, n-1)$,
\item Thief  wins a play if and only if its   final  position $(G',P,n)$ is such that $\scc{G}\neq \emptyset$ and $n=0$.

We define  $\grnk{G}$ to be the minimum $k$ such that  Cops have a winning strategy  in the rank game $\grnk{G,k}$.
\end{definition}

\begin{fact}\label{fact:finite}
There is no infinite play in the game  $\grnk{G,k}$.
\end{fact}
\begin{lemma}\label{Game2forRank:Lemma}
Let $G$ be  a digraph, then Thief has a winning strategy in $\grnk{G,k}$ if and only if he  has a winning strategy in $\Grnk{G,k}$. Therefore $\rnk{G}=\Grnk{G}=\grnk{G}$.
\end{lemma}
\begin{proof}
First,  if  Thief has a winning strategy  in $\Grnk{G,k}$  then he  also has  a winning strategy  in $\grnk{G,k}$  i.e. the latter being without using come back moves.\\
Second, if Thief has a winning strategy in $\grnk{G,k}$  which  uses a come back move of the form $(G',Thief,L,n)\to B$ then he was  able to move early to $B$. 
\cutout{ This would be more easy to see if we consider the game $\mathcal{P}(G,k)$ described in the proof of Fact \ref{fact:finite}. A winning strategy for Thief in $\mathcal{P}(G,k)$ that contains back edges would be transformed into a winning strategy for him that does not contain back edges. }  
\end{proof}

\section{The variable hierarchy}\label{Vari-Hierarchy-Section}
In order to compute the minimum number of bound (i.e. fixed point) variables needed in a $\mu$-term up to $\alpha$-conversion, a digraph measure is required, that is the \emph{entanglement}. \longversion{The variable hierarchy problem was introduced by Immerman and Poizat \cite{Immerman} in the context of  back and forth games $\homl A, B   \homr$.  The aim was of deciding whether  the logical formulae of fixed number of   variables are able or not to distinguish the two structures $A$ and $B$.     The idea is that each player has got a fixed number $p$  of tokens and during the play he mark the position with a token in such a way the number of tokens on both sides is equal. During the game each player can indeed replace the same token. The outcome of the play is as the standard back and forth game apart that only the positions which are marked by tokens are considered.  The existence of a winning strategy for the prover with $p$ tokens ensures that the two structures are models of  formulae of at most $p$ variables. \\
Later,  the variable hierarchy problem was asked for the propositional modal $\mu$-calculus \cite{BerwangerGraLen06} in order to answer the open question whether Parikh Game Logic \cite{Parikh} is a strict subset of modal $\mu$-calculus.  Parikh question was answered affirmatively, Game logic is less expressive than $\mu$-calculus, in \cite{BerwangerLen05} as a consequence of two facts: \emph{(i)} Game Logic is embeddable in the two variable fragment of modal $\mu$-calculus, and \emph{(ii)}  the hierarchy of modal $\mu$-calculus -- made up according to the number of fixed point variables in $\mu$-formulae -- does not collapse. \\
Another consequence of Berwanger's et. all results is the formalization of Immermann's and Poizat's  token games. In other words  Berwanger's et. all  precised that the fact that the minimum number of variables  roughly needed in a $\mu$-formula transfers into a complexity measure  on the underlying digraph of the formula. This complexity measure  is known as the \emph{entanglement}  and it turns out to be the main tool used in analyzing the variables hierarchy problem in $\mu$-calculi. Roughly speaking,  the entanglement of (underlying digraph of) a $\mu$-term,   is the combinatorial  part  of the variable hierarchy. A major consequence is that the comparison between the star height hierarchy and the variable hierarchy transfers, under some assumptions,  into a comparison of the their related combinatorial parts, i.e. between the rank and the entanglement. }
   The entanglement of a finite digraph $G$, denoted $\Ent{G}$, was
defined in \cite{berwanger} by means of some games $\Ent{G,k}$, $k =
0,\ldots ,\card{V_{G}}$. The game $\Ent{G,k}$ is played on
$G$ by Thief against Cops, a team of $k$ cops as follows.
 Initially all the cops are placed outside the digraph, Thief
selects and occupies an initial vertex of $G$.  After Thief's move,
Cops may do nothing, may place a cop from outside the digraph onto the
vertex currently occupied by Thief, may move a cop already on the
graph to the current vertex.  In turn Thief must choose an edge
outgoing from the current vertex whose target is not already occupied
by some cop and move there.  If no such edge exists, then Thief is
caught and Cops win.  Thief wins if he is never caught.  The
entanglement of $G$ is the least $k \in N$ such that $k$ cops have a
strategy to catch the thief on $G$. It will be useful to formalize
these notions.

\begin{definition}\label{entang:def1}
\emph{
  The entanglement game $\Ent{G,k}$ of a digraph $G$ is defined by:
\begin{itemize}
  \item Its positions are of the form $(v,C,P)$, where $v \in V_{G}$,
    $C \subseteq V_{G}$ and $\card{C} \leq k$, $\monespace{2mm} P \in \{Cops,
    Thief\}$.
 \item  Initially Thief chooses $v_{0} \in V_G$ and moves to
    $(v_0,\emptyset,Cops)$. 
 \item Cops can move  from $(v,C,Cops)$ to $(v,C',Thief)$
    where $C'$ can be
     \begin{itemize}
     \item   $C$ : Cops skip,
       \item $C \cup\set{v}$ : Cops add a new Cop on the
      current position,
     \item  $(C \setminus\set{x}) \cup \set{v}$ : Cops move a placed Cop
      to the current position.
\end{itemize}
 \item Thief can move from $(v,C,Thief)$ to $(v',C,Cops)$ if
    $(v,v') \in E_{G}$ and $v' \notin C$.
\end{itemize}
  Every finite play is a win for Cops, and every infinite play is a win
  for Thief. 
}
\end{definition}

\emph{$\Ent{G}$, the entanglement of $G$, is the minimum $k \in \set{
    0,\ldots ,\card{V_{G}}}$ such that Cops have a winning strategy in
  $\Ent{G,k}$}.  \\
Observe that  an undirected path has entanglement at most $2$. The following Proposition, see \cite[\S 3.4]{mathese}, provides a   useful variant of  entanglement games.
\begin{proposition} \label{modif:entag}
 Let $\ET{G,k}$ be the game played as the game $\Ent{G,k}$ apart
  that Cops are  allowed to retire a number of cops placed on the
  digraph. That is, Cops moves are of the form
\begin{itemize}
  \item  $(g,C,Cops) \rightarrow (g,C',Thief)$     (generalized skip move),
  \item  $(g,C,Cops) \rightarrow (g,C'\cup \set{g},Thief)$     (generalized replace move), 
\end{itemize} 
 where in both cases $C' \subseteq C$.
  Then Cops have  a winning strategy in $\Ent{G,k}$ if and only if they 
  have a winning strategy in $\ET{G,k}$.
\end{proposition}

\subsection{An ad  hoc variant of entanglement games}
  The game theoretic definition of the entanglement, Definition  \ref{entang:def1} and even the variant given in Proposition \ref{modif:entag},  refers to some rules which are not very close to the rules of the rank  games with come back  which characterize the rank, Definition  \ref{Rank2-Game-Def}.  And hence we can not  establish the relation between the rank and the entanglement in an easy way. Therefore  we shall give an equivalent variant of entanglement games, denoted $\Entv{G,k}$, with  the property that its rules are close to  those of  the rank games. \\
 First we explain informally the new features of this game w.r.t the games for entanglement. In the game $\Entv{G,k}$ Cops are allowed to  skip, add a cop, replace a cop, retire a number of cops,   and moreover we would like that  they can  put a cop on a  vertex  situated anywhere in the digraph. However, the latter move is not allowed by entanglement rules. In order to make  it possible, Cops should keep in  reserve a set $Vir$ of \emph{virtual}  cops for this purpose: whenever Cops decide to put a cop on an arbitrary vertex $w$  then they should reserve a cop for this purpose and  this cop can not be used until  Thief visits vertex $w$. And at this moment,   the virtual cop must  be placed on $w$.   
\begin{definition}
\emph{
The game $\Entv{G,k}$ is defined as the entanglement game $\ET{G,k}$ apart that  its positions are of the form $(v,C,Vir,P)$ where  $Vir \subseteq V_G$ and $|C \cup Vir|\le k$. Besides the old Cops' moves\footnote{ Which are the skip, the add and the generalized replace.}, the latter act on  the set $C$, Cops can  move   from $(v,C,Vir,Cops)$ to $(v,C',Vir',Thief)$  such that:     
\begin{itemize}
\item  if $v \in Vir$ then Cops must update $C'=C\cup \set{v}$ and $Vir' =Vir \setminus \set{v}$,
\item if $v \notin Vir$  then Cops may update $Vir'=(Vir \setminus A)\cup \set{w} $ where   $w \in V_G$ and  $A\subseteq Vir$. 
\end{itemize}
}
\end{definition}

\begin{lemma}\label{Entag:Vir:Lemma}
Let $G$ be a digraph.   Cops have a winning strategy in $\Entv{G,k}$  if and only if they have a winning strategy in $\ET{G,k}$.
\end{lemma}
\begin{proof}
First,  a winning strategy   for Cops  in $\ET{G,k}$ is still winning for them  in $\Entv{G,k}$, the latter does not refer to  virtual cops. In This case every position $(v,C,Vir,P)$ in $\Entv{G,k}$ is matched with   the position  $(v,C,P)$ in $\Ent{G,k}$ where $Vir=\emptyset$. \\
Second,  a  Cops' winning strategy  in $\Entv{G,k}$ is mapped to a Cops' winning strategy in $\ET{G,k}$ as follows. Every position $(v,C,P)$ of $\ET{G,k}$ is matched with the position $(v,C,Vir,P)$ of $\Entv{G,k}$.\\
A Thief's move from $v$ to $v'$ in $\ET{G,k}$ is simulated by the same move from $v$ to $v'$ in $\Entv{G,k}$. Indeed this simulation  is possible  because Thief is allowed  in $\Entv{G,k}$ to cross  a vertex which is occupied by a virtual cop.\\
Assume that the position $(v,C,Cops)$ is matched with the position $(v,C,Vir,Cops)$ and consider a Cop's move $M=(v,C,Vir,Cops) \to (v,C',Vir',Thief)$ in $\Entv{G,k}$. The move $M$ is simulated in $\ET{G,k}$   by the move  $(v,C,Cops)\to (v,C',Thief)$. 
\end{proof}

\section{The star height hierarchy  vs.  the variable hierarchy}
We are ready to  state the main result of this paper. 
\begin{theorem}\label{Entag:star}
Let $G$ be a digraph. The entanglement of $G$ is lower or equal to the rank of $G$.
\end{theorem}
\begin{proof}
To prove that $\Ent{G}\le \rnk{G}$ it is enough to prove $\Entv{G}\le \grnk{G}$. Because Lemma \ref{Entag:Vir:Lemma} shows that $\Ent{G}=\Entv{G}$ and Lemma \ref{Game2forRank:Lemma} shows that $\rnk{G}=\grnk{G}$.  
 Let $k=\Entv{G}$,  we shall construct a winning strategy for Cops in the  game $\Entv{G,k}$  out of a  Cops' winning strategy in $\grnk{G,k}$.\\
Every position  of the form  $(v,C,Vir,P)$   in  $\Entv{G,k}$  is matched with a position of the form  $(G',P,L,n)$ in $\grnk{G,k}$ such that if  the strongly connected component of  the subgraph $G\setminus (C\cup Vir)$   which contains $v$, denoted by $scc(v)$,   is not trivial then  $scc(v)=G'$.

 Let us consider a  Thief's move in $\Entv{G,k}$  of the form $M=(v,C,Vir,Thief)\to (w,C,Vir,Cops).$
If $scc(w)$ is trivial, then Cops just skip in $\grnk{G,k}$. Observe that either Thief will reach a vertex without  successors (where he loses), or he enters a non trivial strongly connected component. Otherwise i.e.  $scc(w)$ is not trivial, the move $M$ is simulated in $\grnk{G,k}$   according to $w$: 
\begin{itemize}
\item if  $w \in V_{G'}$  then $M$  is simulated  by the move $(G',Thief,L,n)\to (G'',Cops,L,n)$ where  $G''$ is the strongly connected component of $G'$ containing $w$,
\item if  $w \notin V_{G'}$  then the  move $M$  is simulated by the come back move $(G',Thief,L,n) \to (G^{-},Cops,L',m)$  where   $(G^{-},Cops,L',m) \in L$ and  $w \in V_{G^{-}}$.
\end{itemize}
A Cops' move of the form $N=(G',Cops,L,n) \to (G'\setminus w,Thief,L,n-1)$
  in $\grnk{G,k}$ is simulated in $\Entv{G,k}$ according to the nature of the position $(G',Cops,L,n)$.
\begin{itemize}
\item  if the  position $(G'',Thief,L'',n)$ that precedes $(G',Cops,L,n)$  was in the same strongly connected component i.e. $G' \subset G''$ (note that $L=L''$), then the move $N$ is simulated either by    $(v,C,Vir, Cops) \to (v,C\cup \set{v},Vir, Thief)$ if $v=w$, or by    $(v,C,Vir, Cops) \to (v,C, Vir\cup \set{w}, Thief)$ otherwise.
\item  if the  position $(G',Cops,L,n)$  comes from a come back move $(G'',Thief,L'',m) \to (G',Cops,L,n)$ then $N$ is simulated by 
$(v,C,Vir,Cops) \to (v,C\setminus (C \cap V_{G^{-}}),Vir \setminus (Vir \cap V_{G^{-}}),Thief).$\\ Let us define $\minus{G}$.     There exists just one position $\minus{\gamma}$ such that (i) $\minus{\gamma}$ has  the same predecessor of the position $(G',Cops,L,n)$  in the game $\grnk{G,k}$ and (ii) the position $(G'',Thief,L'',m)$ has been reached from the position $\minus{\gamma}$. We define $\minus{G}$ to be the digraph associated to $\minus{\gamma}$, i.e. $\minus{\gamma}$ is of the form $(\minus{G},Cops,\minus{L},\minus{n})$.\end{itemize}
\end{proof}
Finally, besides Theorem  \ref{Entag:star}, we give further conditions under which the non collapse of the variable hierarchy implies the non collapse of the star height hierarchy of a given  $\mu$-calculus $\Lmu$. \\
 The proof schema of the strictness of the variable hierarchy  consists  essentially  in the   construction of \emph{hard} $\mu$-terms of arbitrary entanglement \cite{BerwangerGraLen06,LPAR08}. A $\mu$-term $t$ is said to be hard (w.r.t entanglement) if for every $\mu$-term $t'$ which is equivalent to $t$ we have that the entanglement of $t$ (viewed as a digraph) is lower or equal  to the  entanglement of $t'$ up to a constant.    To argue that the star height of $\Lmu$ is also infinite, it suffices to construct hard $\mu$-terms  of  arbitrary entanglement such that  the entanglement of each $\mu$-term  equals its   rank.  It follows from Theorem \ref{Entag:star} that such  $\mu$-terms  are   also hard w.r.t the star height. \\ 
\textbf{Acknowledgment} We acknowledge helpful discussions with Luigi Santocanale on the topic.

\bibliographystyle{plain}
\bibliography{biblio}

\end{document}